\documentclass[10pt, conference]{IEEEtran}

\usepackage{amsmath,amsthm,amsfonts}
\usepackage{graphicx}
\usepackage{epsf}
\usepackage{mdwmath}
\usepackage{mdwtab}
\usepackage{cite}
\usepackage{amssymb}
\usepackage{algorithmic}
\usepackage{algorithm}
\usepackage{amssymb}
\usepackage{algorithmic}
\usepackage{algorithm}
\usepackage{bm}
\usepackage{color}

\newtheorem{thm}{Theorem}

\theoremstyle{remark}

\newcounter{defn}
\newtheorem{definition}[defn]{Definition}

\newcounter{prop}
\newtheorem{proposition}[prop]{Proposition}

\def\1{\bm{1}}

\begin{document}

\title{Contamination Estimation via Convex Relaxations}

\author{\IEEEauthorblockN{Matthew L. Malloy}
\IEEEauthorblockA{comScore \\ mmalloy@comscore.com}
\and
\IEEEauthorblockN{Scott Alfeld}
\IEEEauthorblockA{University of Wisconsin \\ salfeld@cs.wisc.edu }
\and
\IEEEauthorblockN{Paul Barford}
\IEEEauthorblockA{\hspace{-1cm}  comScore, University of Wisconsin \hspace{-1cm} \\ pb@cs.wisc.edu}}

\maketitle 

\begin{abstract}

Identifying anomalies and contamination in datasets is important in a wide variety of settings.  In this paper, we describe a new technique for estimating contamination in large, discrete valued datasets.  Our approach considers the normal condition of the data to be specified by a model consisting of a set of distributions.  Our key contribution is in our approach to contamination estimation.  Specifically, we develop a technique that identifies the minimum number of data points that must be discarded ({\em i.e.,} the level of contamination) from an empirical data set in order to match the model to within a specified \emph{goodness-of-fit}, controlled by a $p$-value.  Appealing to results from large deviations theory, we show a lower bound on the level of contamination is obtained by solving a series of convex programs. Theoretical results guarantee the bound converges at a rate of $O(\sqrt{\log(p)/p})$, where $p$ is the size of the empirical data set.  

\end{abstract}

\vspace{.1cm}
\normalsize

\section{Introduction}

Anomalies in datasets are typically associated with unexpected or unwanted characteristics such as contamination, noise or outliers that deviate significantly from expectations. The ability to detect anomalies and accurately estimate contamination in datasets is important in a wide variety of domains including healthcare, astronomy, environmental and materials sciences.   
The context that motivates our work is detecting anomalies and estimating contamination in datasets collected from communication and computer systems.  Specific applications of anomaly detection in these datasets include network management and Internet security broadly defined.  Communication and Internet measurement datasets have several distinguishing characteristics including the potential for extreme scale and high dimensionality.

The standard framework for anomaly detection is based on establishing a baseline for {\em normal} ({\em e.g.,} in a distributional sense) and then setting a threshold which if exceeded identifies an anomaly.  The goal in establishing norms and thresholds is to identify anomalies with low false alarm rates.  There is an extensive literature on methods for anomaly detection (see related work in Section~\ref{sec:related}). 

In this paper we describe a new method for anomaly detection which is based on estimating the level of contamination in a dataset.  An anomaly is declared if a dataset has an elevated level of contaminate. We consider the contamination-free ({\em i.e.,} normal) condition of a dataset to be specified by a model comprised of a set of distributions.  We then compare the model to the distributional profile of a target dataset collected over a specified period.  A standard method for comparing datasets in this way is goodness of fit (GoF) testing~\cite{d1986goodness}.  
To the best of our knowledge, this paper is the first to address the problem of contamination estimation using GoF testing based on entropy minimization, as we define in Section \ref{sec:probstate}.  

The approach we develop is based on answering the following question.  Given a model consisting of a family of distributions, a specified $p$-value, and an empirical dataset, what is the minimum number of data points that must be discarded so that the empirical distribution of the data matches a member model distribution (in terms of GoF for a specified $p$-value)?  This is akin to finding the largest subset of the original dataset which has an empirical distribution \emph{close} to the model. 
We show that this question can be efficiently answered by solving a series of convex optimizations.  Solving the optimizations results in a lower bound on the minimum number of data points that are attributed to a contaminate.   In the simplest case, each convex optimization is an inequality constrained entropy minimization problem (whose dual is a constrained geometric program) which can be solved in real time and at scale for many applications.  More generally, the approach can be applied to any setting in which the model consists of a convex set of distributions.  Two specific instances which we discuss are \emph{1)} models defined by any number of distributions with arbitrary mixture proportions, and \emph{2)} models defined by the set distributions with small Kullback-Leibler (KL) divergence to a specified distribution, which arises when the model itself is generated from a finite amount of data.   Lastly, we show the lower bound output by the optimization converges to an upper bound known as the separation distance at a rate of $O(\sqrt{\log ( p)/p  })$, where $p$ is the number of data points.

\section{Quantifying Contamination} \label{sec:prob_statement}

\subsection{Notation}

Let $P \in  \mathbb{R}^n$ and $Q \in \mathbb{R}^n$ denote probability mass functions over $n$ categories, with elements $P_i$, $i=1,\dots,n$ and $Q_i$, $i=1,\dots,n$.  Throughout, $P$ denotes the distribution under test, $Q$ denotes a member distribution of the model, $Q^0$ denotes the `true' unknown model distribution, and $Q^j$ indexes multiple distributions. The empirical distribution of a sequence of random variables $X=X_1, \dots, X_p \in \mathcal{X}^p$ is the relative proportion of occurrences of each element of $\mathcal{X}$ in $X$.  Specifically, let $\mathcal{X}=: \left\{x_1, x_2, \dots, x_n \right\}$ and define $p_i = \sum_{j=1}^p \mathbf{1}_{\left\{X_j=x_i \right\}} $ for $i=1,\dots,n$.  Then 
$\widehat{P}(X) =  \frac{1}{p}\left\{ p_1, p_2, \dots, p_n \right\}.$  
$\mathbb{P}_{Q}(\cdot)$ denotes probability measure with respect to distribution $Q$.  For simplicity of notation, we write
$\mathbb{P}_{{Q}}(\{\widehat{P}^1,\widehat{P}^2\})$ as short hand for $\mathbb{P}_{{Q}}\left(\left\{X \in \mathcal{X}^p :  \widehat{P}(X) \in \{ \widehat{P}^1,\widehat{P}^2 \} \right\} \right)$. 
The Kullback-Leibler divergence between two distributions is defined in the usual manner,
\begin{align*}
D(P||Q) := \sum_{i} P_i \log \left( \frac{P_i}{Q_i}\right).
\end{align*}
$D(P||Q)$ is a jointly convex function in $P$ and $Q$.  The minimum entropy set, $\left\{P: D(P||Q) \leq \epsilon \right\}$, is a convex set (for a fixed $Q$, $\epsilon$).   Lastly, let $\mathbb{S}^n$ denote the probability simplex:
\begin{align} \nonumber
\mathbb{S}^n := \left\{P \in \mathbb{R}^n: \sum_i P_i =1, \ P_i \geq 0 \quad i=1, \dots, n \right\}.
\end{align}

\subsection{Quantifying Contamination} \label{sec:probstate}

Consider a set of model distributions $\mathcal{Q}$ whose elements are supported over a finite number of categories $\mathcal{X}$ with $|\mathcal{X}|=n$.  For example, $\mathcal{Q}$ could be set of minimum entropy distributions, or a mixture distribution, $Q = \sum_{j=1}^\ell \pi_j Q^{j}$, 
where $\pi_1,\dots,\pi_{\ell}$ are unknown ($\mathcal{Q}$ is the set of all such mixture distributions). Let $X \in \mathcal{X}^p$ denote a collection of samples.  An unknown subset of the samples consists of $i.i.d.$ draws from an unknown distribution $Q \in \mathcal{Q}$.     The remaining samples, $\mathcal{C} \subset \lbrack p \rbrack $, are generated by some other means, and correspond to \emph{contaminated} samples.    
This paper is concerned with lower bounding the size of the contaminating set $\mathcal{C}$ given the set of model distributions $\mathcal{Q}$, a specified significance level (a $p$-value), and the observed samples $X_1,\dots, X_p$.

Intuitively, if the empirical distribution of a sequence of random variables is \emph{close} to the model distribution in terms of GoF, we conclude the sequence is \emph{not} contaminated. To quantify this intuition, we define a set of \emph{typical} empirical distributions based on statistical significance; we note this definition is distinct from the usual definitions of \emph{strongly} and \emph{weakly} typical, and making this connection is a contribution herein. 

\begin{definition}{Typical}.  Let $\widehat{P}^{1}, \widehat{P}^{2}, \dots$ be an ordering on all empirical distributions (of $p$ samples and $n$ categories) such that $\mathbb{P}_{{Q}}(\widehat{P}^{1}) \leq \mathbb{P}_{Q}(\widehat{P}^{2}) \leq \dots$.   A sequence of random variables $X$ with $\widehat{P}(X) = \widehat{P}^{\ell}$ is \emph{typical} at significance level $\epsilon$ with respect to $\mathcal{Q}$  iff
\begin{align} \label{eqn:atypical}
\sup_{Q \in \mathcal{Q}}  \mathbb{P}_Q\left(\left\{\widehat{P}^{1}, \widehat{P}^{2}, \dots, \widehat{P}^{\ell-1}, \widehat{P}^{\ell}\right\}\right) \geq \epsilon
\end{align}
for any such ordering\footnote{Note the ordering is an implicit function of $Q$; we suppress this for simplicity of notation.}.
\end{definition}
The definition implies a sequence of random variables $X$ is typical if the probability of the empirical distribution of $X$ \emph{or any less likely empirical distribution} is more than a specified significance level.
Note $\epsilon$ is interpreted as a $p$-value; as $\epsilon$ approaches zero, all sequences become typical (requiring stronger evidence to reject the null hypothesis).  As $\epsilon$ increases, fewer sequences are typical.

\begin{definition}{Contaminated.} 
We say $X$ is \emph{contaminated} iff $X$ is not typical (with respect to $\mathcal{Q}$ and with significance $\epsilon$).  Likewise, an empirical distribution $\widehat{P}(X)$ is \emph{contaminated} iff $X$ is not typical.
\end{definition}

In this paper we study the following question.  Let $X = X_1,\dots, X_p$ be a dataset, and let $X_{\widehat{\mathcal{C} }}= \{X_i :  i \in \widehat{ \mathcal{C}} \ \}$ be any subset of of the original dataset.   What is the smallest set $\widehat{\mathcal{C}} \subset \lbrack p \rbrack $ such that $X_{\lbrack p \rbrack  \setminus \widehat{\mathcal{C}}}$ \emph{is not contaminated}?  Specifically, let
\begin{align*}
c^* = \inf  \left\{ \vert \widehat{\mathcal{C} }   \vert :  X_{\lbrack p \rbrack  \setminus \widehat{\mathcal{C}}} \mbox{ is typical for $(\mathcal{Q}, \epsilon)$}  \right\}.
\end{align*} 
How and under what conditions can one compute $c^*$ efficiently?  Our main focus and insight will be on the continuous approximation to $c^*/p$, denoted $\alpha^*$:
\begin{align*}
\alpha^* = \inf  \left\{ \alpha \in \lbrack 0,1 \rbrack : \exists P \in \mathcal{P}(X,\alpha) \mbox{ typical for $(\mathcal{Q}, \epsilon)$}  \right\}
\end{align*} 
where $\mathcal{P}(X,\alpha)$ is the set of all distributions that can be created by discarding a fraction $\alpha$ of the mass of $\widehat{P}(X)$ (see Sec. \ref{sec:ConvRel}):
\begin{align} \label{eqn:pram}
\mathcal{P}(X,\alpha) = \left\{P \in \mathbb{S}^n :P_i \leq \frac{\widehat{P}_i(X)}{1-\alpha}  \quad  i=1,\dots,n \right\}.
\end{align}
Throughout, $\alpha$ is a key parameter that represents the fraction of the dataset attributed to contamination; $\alpha^*$ represents the smallest $\alpha$ such that there exists a subset of the original data of size $p(1- \alpha)$ that is \emph{not} contaminated.  If $\alpha^* = 0$, the original dataset is not contaminated; if $\alpha^* = 1$, the entire dataset must be attributed to contamination.


 \subsection{Separation Distance}  
We assume $X_i \overset{i.i.d.} \sim Q^{0}$ for all $i \not \in \mathcal{C}$.  For $X_i$, $i \in \mathcal{C}$, no assumption is made.  This agnostic approach has inherent limitations.  In the extreme case the distribution of the contaminated data could exactly follow that of the model. Here, the distribution of the full dataset should closely match the model, and be indistinguishable from the setting where $\mathcal{C}$ is empty.   No contamination should be reported to within the significance level (in $m$ realizations of $X^p$, we expect $c^* \neq  0$ fewer than $m \epsilon$ times).

A more interesting scenario is when the empirical distribution of the full dataset converges to a \emph{distinct} distribution \emph{i.e.}, $\widehat{P}(X^p) \rightarrow P \neq Q^{0}$.  In the case that $\mathcal{Q} = \left\{ Q^0\right\}$, a consistent estimator will report non-zero contamination for large $p$.  $P$ can be written as a mixture distribution, and we are interested in reporting the smallest $\kappa$ such that $(1-\kappa) Q^{0} + \kappa F = P$
for \emph{any} distribution $F$.  $F$ represents the contaminating distribution, and $\kappa$ the proportion of the samples which are drawn from $F$.   This minimum value of $\kappa$ is known as the \emph{separation distance}~\cite{aldous1987strong} between $P$ and $Q^{0}$, written succinctly as 
\begin{align*}
\kappa(P||Q^{0}) = \max_{i \in [n]} \left(1 -  \frac{P_i}{Q_i^{0}}\right).
\end{align*}
In this way, the separation distance between the empirical distribution of the data and model distribution plays an important role in the behavior of $c^*$ and $\alpha^*$ as the sample size grows.   We show as a corollary to later results that $\alpha^*$ is both upper bounded by and converges to $\kappa(\widehat{P}(X)||Q^{0})$  as $p$ grows (see Proposition \ref{prop:12324}  and Theorem \ref{thm:largep}).

 \subsection{Convex Relaxations} \label{sec:ConvRel}
With the exception of problems involving data over only two categories ($n=2$), directly checking if a sample is contaminated is computationally prohibitive, even in the setting where the model consists of a single distribution (when $\mathcal{Q} = \{Q^{0}\}$). Alternatively, using large deviations results, bounds can be derived.  The bound presented below can confirm if a particular dataset is contaminated. The theorem involves the KL divergence between the empirical distribution and a member of $\mathcal{Q}$.  In the case where $\mathcal{Q} = \{Q^{0}\}$, the bound provides a simple way to check if a sample is contaminated at a particular significance level $\epsilon$; in the more general case, if $\mathcal{Q}$ is a convex set, numerical optimization techniques can efficiently check the condition.
 
 \begin{thm}(Outer Bound). \label{thm:outer}
If
\begin{align} \label{eqn:Qouter}
 \inf_{Q \in \mathcal{Q} } D(\widehat{P}(X)||Q)  \geq \frac{1}{p} \log \left( \frac{1}{\epsilon} \right) + \frac{2n}{p} \log (p+1) 
\end{align}
then $X$  is contaminated at significance level $\epsilon$.
\end{thm}
\begin{proof} See Appendix A.
\end{proof}
Theorem \ref{thm:outer} is an outer bound; any empirical distribution with KL distance \emph{greater} than the stated quantity (from \emph{all} elements in $\mathcal{Q}$) \emph{is} contaminated.  
Theorem \ref{thm:outer} can be used to bound the size of the smallest set $\mathcal{C} \subset \lbrack p \rbrack $ such that $X_{\lbrack p \rbrack  \setminus \mathcal{C}}$ \emph{is not contaminated}. This is simplified if $\mathcal{Q}$ consists of a single model distribution; we first discuss this scenario.  In principle, given a dataset $X \in \mathcal{X}^p$ and a model distribution $Q^{0}$, one could first check if $X$ is contaminated by evaluating (\ref{eqn:Qouter}).  If (\ref{eqn:Qouter}) holds, $X$ is contaminated, and an immediate question follows -- how many and which data points must be excluded so that (\ref{eqn:Qouter}) no longer holds?   A exhaustive approach to answer this question would be the following.  For each $x_i \in \mathcal{X}$, discard a single data point that takes the value $x_i$, and recalculate the empirical distribution with the data point removed.  Of the $n$ new empirical distributions, check if the one with minimum KL divergence to the model distribution still satisfies (\ref{eqn:Qouter}). 

If (\ref{eqn:Qouter}) still holds for all possible empirical distributions with one data point removed, check all distinct empirical distributions that can be created by discarding 2 data points (roughly $n^2$ possibilities, provided each $x_i$ appears at least twice in the data).  Continuing in this manner, one would check each of the $\sim n^m$ possible empirical distributions that can be created by discarding $m$ data points.  When (\ref{eqn:Qouter}) is first violated, $m$ lower bounds the minimum number of data points that must be excluded to match the model.  We can interpret this as a series of integer programs.  For $m=0,\dots, p$ define $D^*_m$ as the solution to  
\begin{equation}
\begin{aligned} \label{eqn:int}
& \underset{m_1, m_2,\dots,m_n \in \mathbb{N}^n}{\text{minimize}}
& & \sum_{i=1}^n \frac{p_i - m_i}{p-m}  \log \left(\frac{\frac{p_i - m_i}{p-m}}{Q^{0}_i}  \right) \\
& \text{subject to} & & \sum_i m_i = m  \\
& & & m_i \leq p_i \qquad  i = 1, \ldots, n
\end{aligned}
\end{equation}
where $p_i$ is the number of times $x_i$ appears in the original dataset $X$. The optimization variables, $m_i$, represent the number of samples to discard corresponding to a particular $x_i$.   Note that the objective is the KL divergence between the \emph{new} empirical distribution (with $m$ samples removed) and the known distribution $Q^{0}$.   
The value of $D^*_m$ can be checked in Theorem \ref{thm:outer}, providing conditions under which one can find a set $|\mathcal{C}| =m$ such that $X_{\lbrack p \rbrack  \setminus \mathcal{C}}$ is not contaminated.  This gives a bound on $c^*$.  Specifically, 
\begin{align*} 
c^* \geq  \max \left\{m: D^*_m \geq \frac{1}{p-m} \log \left( \frac{1}{\epsilon} \right) \right. \hspace{1cm} \\
\left. + \frac{2n}{p-m} \log (p-m+1)  \right\}. \hspace{-.9cm}  \nonumber
\end{align*}
Note that the condition in Theorem \ref{thm:outer} will always be met for some $m$; in particular, for $m=p$, by convention $D_0^*=0$, implying that the empty set, $X_{ \{ \} }$, is not contaminated. 

The optimization in (\ref{eqn:int}) is an integer program over a subset of $\mathbb{N}^n$.  
To efficiently solve the optimization, we can translate the integer valued variables to their continuous counterparts; specifically, let $\widehat{P}_i=  p_i/p$, be the original empirical distribution, and  $\alpha = m/p$ represent the fraction the total samples discarded.  Making these substitutions results in a convex entropy minimization problem:
\begin{equation} \label{eqn:contop1}
\begin{aligned} 
& \underset{P \in \mathbb{S}^n}{\text{minimize}}
& & \sum_i P_i  \log \left(\frac{P_i}{Q_i^{0}} \right) \\
& \text{subject to} & &   P_i \leq \frac{\widehat{P}_i}{1-\alpha}  \qquad  i = 1, \ldots, n 
\end{aligned}
\end{equation}
where $\alpha \in \lbrack 0,1 \rbrack$ represents the fraction of samples removed. 

More generally, $\mathcal{Q}$ is a set of distributions.  The same continuous approximation results in a joint optimization over the model space $\mathcal{Q}$ and the space of empirical distributions, $\mathcal{P}(X,\alpha)$ defined in (\ref{eqn:pram}).   Formally, let $D^*_\alpha$ be given as
\begin{equation} \label{eqn:contop}
\begin{aligned} 
D^*_\alpha = & \underset{P \in \mathcal{P}(X,\alpha), Q \in \mathcal{Q}}{\text{min}}
& & \sum_i P_i  \log \left(\frac{P_i}{Q_i} \right). \\
\end{aligned}
\end{equation}
If  $\mathcal{Q}$ is a convex set, the above optimization can be efficiently solved in many settings (see Sec. \ref{sec:minentro}).

To answer our original question and bound $\alpha^*$, one can conduct a line search over $\alpha \in \lbrack 0,1 \rbrack$, repeatedly solving the above optimization, and checking the output value of $D^*_\alpha$ against Theorem \ref{thm:outer}. 
This is captured in the following proposition. 
\begin{proposition} \label{prop:12324}
Let
\begin{align} \label{eqn:alp_ub}
\alpha_{\mathrm{L}} =  \max \left\{\alpha: D^*_\alpha \geq \frac{1}{p(1-\alpha)} \log \left( \frac{1}{\epsilon} \right) \right.   \hspace{1cm}  \\
 \left. + \frac{2n}{p(1-\alpha)} \log \left(p(1-\alpha)+1\right)  \right\}   \hspace{-.3cm}   \nonumber
\end{align}
then $\alpha_{\mathrm{L}}\leq \alpha^*$.
\end{proposition}
\begin{proof}
The proof follows directly from Theorem \ref{thm:outer}.  For any $\alpha$ such that the condition on $D^*_\alpha$ in (\ref{eqn:alp_ub}) holds, by Theorem \ref{thm:outer}, any distribution in $\mathcal{P}(X,\alpha)$ is contaminated.   We note that $\alpha_{\mathrm{L}}$ always exists by monotone properties of $D_\alpha^*$ and the right hand side of the conditional in (\ref{eqn:alp_ub}).  See Appendix B, Theorem \ref{thm:largep} for details.  
\end{proof}

Fig. \ref{fig:geo} shows a geometric interpretation of Proposition \ref{prop:12324} and the optimization in (\ref{eqn:contop1}).  See the caption for details.



The lower bound obtained by solving the series of optimization problems converges to the separation distance, captured by the following theorem.
\begin{thm}  \label{thm:conv} Let $\mathcal{Q} = \{ Q^{0} \}$.  Fix $\widehat{P}(X)$.  Then    
 \begin{align*}
\kappa(\widehat{P}||Q^{0}) -  \alpha_{\mathrm{L}} = O\left(\sqrt{\frac{\log p}{p} } \right).
 \end{align*}
\end{thm}
\begin{proof}
See Appendix B.
\end{proof} 
Theorem \ref{thm:conv} is stated for a fixed $\widehat{P}(X)$, although one would in general assume $\widehat{P}(X)$ to be an implicit function of $p$.  The reason for fixing $\widehat{P}(X)$ is both generality and simplicity.  The assumption decouples randomness from the convergence rate of the upper bound and the lower bound produced the optimization; without this assumption, the upper and lower bounds would be random variables, and necessitate a probabilistic statement.  
We also note that a precise limit statement can be readily extracted from the proof.

\begin{figure}
\centering
\includegraphics[width=.9\columnwidth]{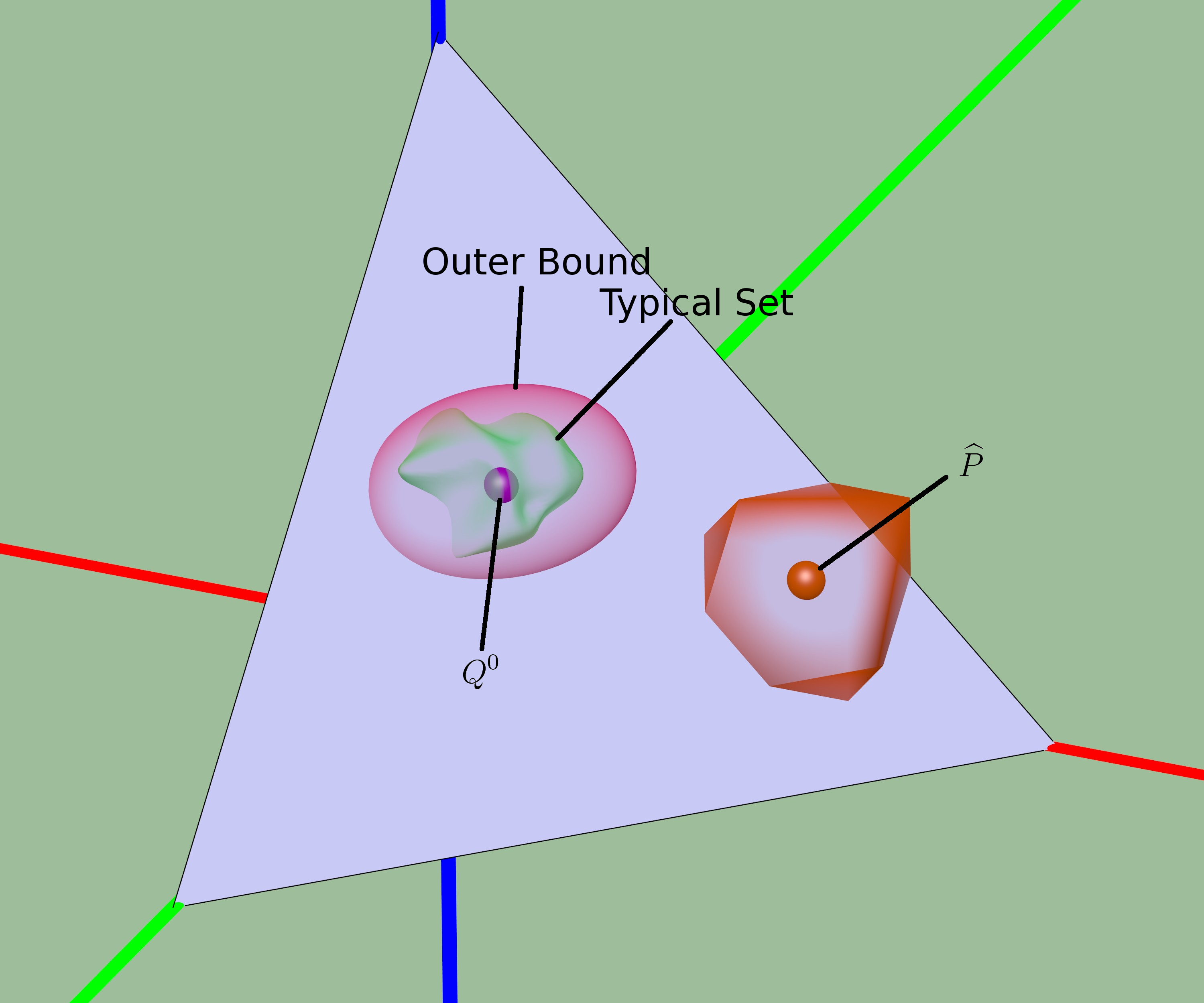}
 \caption{Geometric interpretation of Proposition \ref{prop:12324} and the optimization in (\ref{eqn:contop1}) with $\mathcal{Q}= \left\{ Q^0\right\}$. The width of the hypercube around $\widehat{P}$ is $\alpha$.  As $\alpha$ is increased, the hypercube eventually intersects the `outer bound' set, which represents the set of distributions closest to ${Q}^{0}$ in KL divergence; the sets intersect when $\alpha =\alpha_{\mathrm{L}}$.     Note that the `outer' bound set also increases in size as $\alpha$ increases.    \label{fig:geo}}
  \label{fig:3d}
\end{figure}

\subsection{Discussion} \label{sec:mix} \label{sec:minentro}
In practice, it is often the case that the precise model distribution is not known; instead, it may be known that the model distribution comes from some family of distributions.  This arises in anomaly detection when normal events are known to correspond to unknown proportions of samples from a finite set of distributions.  This is the case of the mixture model \emph{i.e.}, $\mathcal{Q}$ is the set of all distributions that can be represented as $ Q = \sum \pi_j Q^{j}$ for any mixture proportions $\pi_j$.  As the set of mixture distributions with unknown mixture components is a convex set, we can directly address this setting using the developments of Sec. \ref{sec:ConvRel}.
Jointly optimizing over the mixture weights and the mixture distribution, the optimization takes the form 
\begin{equation}
\begin{aligned} \label{eqn:mixmod}
& \underset{P \in \mathcal{P}^n, \ \pi \in \mathbb{S}^k }{\text{minimize}}
& & \sum_i P_i  \log \left(\frac{P_i}{\sum_{j=1}^k \pi_j Q_{i}^{j}} \right). \\
\end{aligned}
\end{equation}
We note that the above optimization can be solved at scale in real time for many applications; see discussions of numerical experiments below for details.


For many applications, model distributions are generated using a \emph{finite} amount of data from known good sources (\emph{i.e.}, sources that are known to have no contamination).   Let $\widehat{Q}$ be an empirical distribution generated from $p'$ samples of an $i.i.d.$ population, and consider the set 
\begin{eqnarray*}
\mathcal{Q}' = \left\{Q: \widehat{Q} \mbox{ is typical for } (\{Q\}, \epsilon) \right\}.
\end{eqnarray*}
Here, $\mathcal{Q}$ is the set of all distributions that have $\widehat{Q}$ as a typical empirical distribution.  As before, determining membership in $\mathcal{Q}$ is intractable for large $p'$ and more than two categories.   Let 
\begin{align*} 
\bar{\mathcal{Q}} = \left\{Q : D(\widehat{Q}||Q) \leq \frac{1}{p'} \log \left(\frac{1}{\epsilon} \right) + \frac{2n}{p'} \log(p'+1)\right\}.
\end{align*}
$\bar{\mathcal{Q}}$ satisfies two important properties.  First, $\mathcal{Q} \subseteq \bar{\mathcal{Q}}$ by Theorem \ref{thm:outer} and second, $\bar{\mathcal{Q}}$ is a convex set. 

Solving the optimization in (\ref{eqn:contop}) with $\mathcal{Q} = \bar{\mathcal{Q}}$ provides a powerful result which we state in the following proposition. 
\begin{proposition} \label{prob:last}
Consider two empirical distributions $\widehat{P}$ and $\widehat{Q}$.  Let $\mathcal{Q} = \bar{\mathcal{Q}}$, defined above, and let $D_0^*$ be the solution to the optimization in (\ref{eqn:contop}) with $\alpha = 0$.  If
\begin{align*}
D^*_0 \geq  \frac{1}{p} \log \left( \frac{1}{\epsilon} \right) + \frac{2n}{p} \log \left(p+1\right),
\end{align*}
there is no $Q$ that simultaneously satisfies \emph{1)} $\widehat{Q}$ is typical with respect to $Q$
and \emph{2)} $\widehat{P}$ is typical with respect to $Q$. 
\end{proposition}
Satisfying proposition \ref{prob:last} implies that observing a $\widehat{Q}$ and a $\widehat{P}$ generated by the same underlying distribution \emph{by chance} can occur at most a fraction $\epsilon$ of the time;  in this sense, $\widehat{P}$  must be contaminated.  With a single parameter search over $\alpha \in \lbrack 0,1 \rbrack$, the lower bound applies: $\alpha^* \geq \alpha_{\mathrm{L}}$.  We note that the formulation does not require the empirical model and the distribution under test to have joint support.



\begin{figure}
  \centering
  \includegraphics[width=.91\columnwidth]{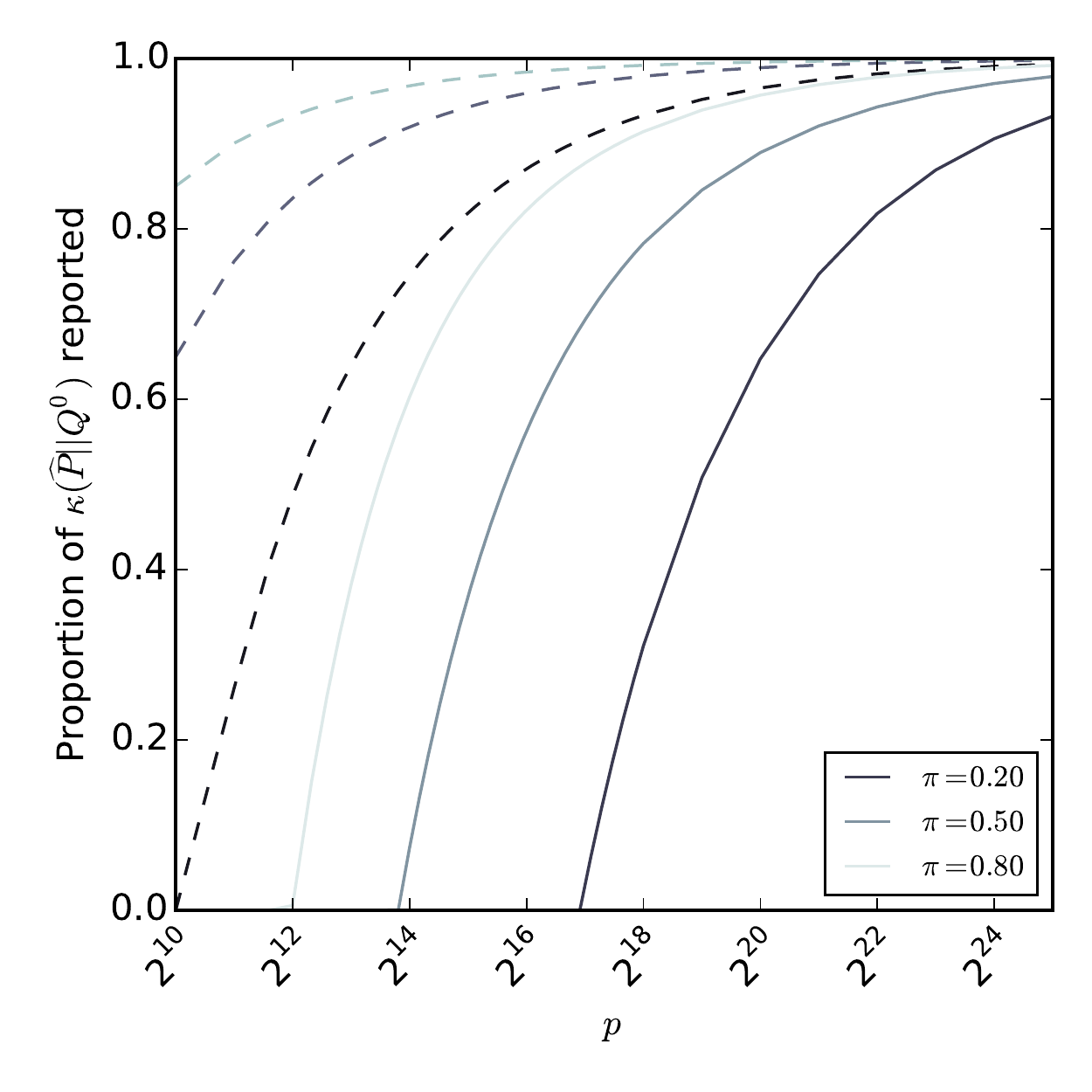}
  \caption{Numerical example. $n=11$, $\epsilon = 0.05$, \(\mathcal{Q} = \{Q^0\} \), with $Q^0$ a uniform distribution over 11 categories.  Solid lines show \(\alpha_\mathrm{L} \) divided by \( \kappa(\widehat{P}||Q^0) \) for mixture distributions \(\widehat P_{\textrm{dip}} = (1 - \pi) \ Q^0 + \pi \mathcal{U}_{10}\), where $\mathcal{U}_{10}$ is a uniform distribution over 10 of the 11 categories.  Dashed lines show 
  \(\alpha_L \) divided by \( \kappa(\widehat{P}||Q^0) \)  for    
\(\widehat P_\textrm{spike} = (1 - \pi)  Q^0 + \pi  \delta\), where $\delta$ is a point mass.  
}
  \label{fig:varyp_normed}
\end{figure}

Numerical experiments were conducted to highlight the utility of Proposition \ref{prop:12324}; results are shown in Fig. \ref{fig:varyp_normed}.  In contrast to the deterministic experiments in Fig. \ref{fig:varyp_normed}, experiments with random samples from various model and test distributions as input were run, showing similar convergence behavior. An experiment with with $\mathcal{Q}$ being a set of 10 mixture distributions with $n=50$ was also conducted.  The line search over $\alpha$ was completed using a bisecting search to an accuracy of $2^{-28}$ (the optimization was solved 27 times for each experiment).  Averaged over 50 trials, the total time to compute $\alpha_{\mathrm{L}}$ was 0.4 seconds.  Experiments were implemented using CVXOPT  \cite{cvxopt} and results visualized with matplotlib \cite{Hunter:2007}.

\section{Related Work} \label{sec:related}
Related work can be broadly classified into traditional work in goodness of fit (GoF) testing, and more recent work in anomaly detection.  
GoF testing has an extensive literature.  When the data are binary valued, and the model distribution Bernoulli,  quantifying contamination using GoF tests can be addressed by evaluating binomial probabilities (a technique known as Fisher's Exact method \cite{mehta1984exact}).  When the data take on more than two values, exact solutions for the level of contamination become intractable. 

A customary approach to GoF testing for categorical data is Pearson's $\chi^2$ test \cite{agresti2014categorical}.  This approach to GoF testing can be quite powerful, but suffers from limitations.  $\chi^2$ tests are approximations, and are known to be invalid under certain conditions.  In particular, the test is invalid when $p_i = 0$ for one or more categories.   Nonetheless, employing the $\chi^2$ test, one can deduce another optimization (much as we do in Sec. \ref{sec:ConvRel}) to answer the aforementioned question; we note the resulting optimization is a separable quadratic program with linear equality constraints which has an analytic solution \cite{bay2010analytic}, and would be an interesting starting point for future work.   Since Pearson's $\chi^2$ test hinges on a normal approximation, this approach would not result in strict contamination bounds.  
More specific to the contamination estimation problem presented here, recent work includes decontamination with multiclass label noise \cite{blanchard2014decontamination, scott2013classification}, which focuses on recovering proportions of a set of mixture distributions present in dataset.     

There is an extensive literature on the related topics of anomaly detection and outlier detection including work employing entropy based techniques, in particular \cite{hero2006geometric} and \cite{gu2005detecting}; we note the formulations here are distinct in that the level of contamination is not estimated.  Lastly, we briefly discuss related work in anomaly detection the areas of computer networks, systems and security as this is the motivation for our developments. Early work on identifying anomalous or unexpected behaviors such faults ({\em e.g.,} due to outages or failures) or spikes ({\em e.g.,} associated with DoS attacks or flash crowds) in computer network traffic was based on the application of graph models, time series and multi-resolution methods {\em e.g.,}~\cite{Feather93,Katzela95,Brutlag00,Barford02}, and Principle Components Analysis (PCA)~\cite{Lakina04,Lakina04a,Lakina05}.  There are significant difficulties in tuning these methods to provide low false alarm rates in practice~\cite{Ringberg07}, necessitating methods based on statistical significance, as presented here.

\bibliographystyle{IEEEtran}
\bibliography{InvalidSampleDetection.bib}

\onecolumn

\section*{Appendix A}
 Proof of Theorem \ref{thm:outer} requires two ingredients, both relying on results from large deviations theory.  The first ingredient is Sanov's Theorem, which we state below.  
\begin{thm}  \label{thm:sanov} (Sanov's Theorem) \cite{cover2012elements} (Theorem 11.4.1).  Let $\mathcal{S}$ be a set of empirical distributions (with $p$ samples over $n$ categories).  Then
\begin{align}
\mathbb{P}_Q(\mathcal{S}) \leq  (p+1)^n \exp\left(-p  \min_{\widehat{P} \in \mathcal{S}}  D(\widehat{P}||Q)\right).
\end{align}
\end{thm}
The second ingredient is also readily derived from results in large deviations theory. 
\begin{thm} \label{thm:bnd1}
Let $\mathcal{S}$ be a set of empirical distributions such that $\mathbb{P}_{{Q}}(\widehat{P}^\ell) \geq \mathbb{P}_{{Q}}(\widehat{P})$ for all $\widehat{P} \in \mathcal{S}$.
Then, 
\begin{align*}
\min_{ \widehat{P} \in \mathcal{S} }  D(\widehat{P}||Q)
 \geq D(\widehat{P}^{\ell}||Q) - \frac{n}{p} \log (p+1).
\end{align*}
\end{thm}
\begin{proof}
The following inequalities hold \cite{cover2012elements} (Theorem 11.1.4):
\begin{align} \label{eqn:typebnd} 
\frac{1}{(p+1)^n} \exp\left(-p D(\widehat{P}||Q)\right) \leq \mathbb{P}_Q(\widehat{P}) \leq \exp \left(-p D(\widehat{P}||Q) \right).
\end{align}
Thus, for any $\mathbb{P}_Q(\widehat{P}^{m}) \leq \mathbb{P}_Q (\widehat{P}^{\ell})$, 
\begin{align*}
\frac{1}{(p+1)^n} \exp\left(-p D(\widehat{P}^{m}||Q)\right) \leq \exp\left(-p D(\widehat{P}^{\ell}||Q)\right)
\end{align*}
which implies the result, 
completing the proof of Theorem \ref{thm:bnd1}.
\end{proof}

Combining Theorems \ref{thm:sanov}  and \ref{thm:bnd1}, we have 
\begin{align*}
\mathbb{P}_Q(\{ \widehat{P}^1, \widehat{P}^2, \dots, \widehat{P}^\ell \} ) \leq  (p+1)^{2n} \exp\left(-p D(\widehat{P}^\ell ||Q)   \right)
\end{align*}
provided $\mathbb{P}_{{Q}}(\widehat{P}^1) \leq \mathbb{P}_{Q}(\widehat{P}^2) \leq \dots \leq \mathbb{P}_{Q}( \widehat{P}^{\ell})$.
This provides a simple way to confirm if a sample is \emph{contaminated} at a particular significance level $\epsilon$.  In particular, assume $\widehat{P}(X) = \widehat{P}^{\ell}$.  If
\begin{align*} 
 (p+1)^{2n} \exp\left(-p D(\widehat{P}(X)||Q)  \right) \leq \epsilon 
\end{align*}
or equivalently
\begin{align} \label{eqn:4454e}
 D(\widehat{P}(X)||Q)  \geq \frac{1}{p} \log \left( \frac{1}{\epsilon} \right) + \frac{2n}{p} \log (p+1) 
\end{align}
then $X$ is not typical; $X$ satisfies 
\begin{align*}
\mathbb{P}_Q( \{ \widehat{P}^1, \widehat{P}^2, \dots, \widehat{P}^\ell \} ) \leq  \epsilon
\end{align*}
and is contaminated with significance  $ \epsilon$.  If (\ref{eqn:4454e}) holds for all $Q \in \mathcal{Q}$, in other words, if 
\begin{eqnarray*}
\inf_{ Q \in \mathcal{Q} } D(\widehat{P}(X)||Q) \geq  \frac{1}{p} \log \left( \frac{1}{\epsilon} \right) + \frac{2n}{p} \log (p+1)
\end{eqnarray*}
we conclude then $X$ is not typical with respect to $(\mathcal{Q},\epsilon)$, implying the result.

\section*{Appendix B} \label{app:C}
Proof of Theorem \ref{thm:conv}.  The proof requires three main steps.    The first step is to show that when $\alpha$ is sufficiently close to $\kappa(\widehat{P}||Q^{0})$, the solution to (\ref{eqn:contop}) can be written in closed form.  The second step is to show a number of properties regarding the asymptotic behavior of $\alpha_{\mathrm{L}}$ as $p$ grows; specifically, $\alpha_{\mathrm{L}}$ is monotone increasing in $p$, and converges to the separation distance; these properties imply that for large $p$, the closed form solution is valid.  Lastly, we can bound the difference between $\kappa(\widehat{P}||Q^{0})$  and $\alpha_{\mathrm{L}}$ using the closed form solution.


{\bf{Step 1:}} For $\alpha$ close to the separation distance (equivalently, for large $p$, as we show next in Theorem \ref{thm:largep}), the optimization has a closed form.  This is captured in the following Theorem.  Note the theorem assumes there is a unique largest  $\frac{\widehat{P}_\ell}{Q_\ell}$; in the degenerate case when this is not true, the theorem can be restated introducing at most a factor of $n$, which does not affect the final result.  
\begin{thm}
Let $\frac{\widehat{P}_i}{Q_i}$ be ordered such that $\frac{\widehat{P}_\ell}{Q_\ell} <  \frac{\widehat{P}_k}{Q_k} \leq \dots \leq \frac{\widehat{P}_n}{Q_n} $.  For $\alpha \in \lbrack 1 - \widehat{P}_\ell - \frac{\widehat{P}_k}{Q_k} \left(1 - Q_\ell   \right)    , \ \kappa(\widehat{P}||Q) \rbrack$ 
\begin{align} \label{eqn:soon}
P_i^* = \begin{cases} 
\frac{Q_i \left(1-  \frac{\widehat{P}_{\ell}}{1-\alpha} \right)}{ 1 - Q_\ell  }   \qquad & i \neq \ell \\
\frac{\widehat{P}_{\ell}}{1-\alpha}  \qquad&  i = \ell
\end{cases}
\end{align}
is the unique solution to (\ref{eqn:contop}).
\end{thm}

\begin{proof}
The result can be shown by verifying the conditions KKT conditions with
\begin{align} \label{eqn:lamstarer}
\lambda_i^* = \begin{cases}  0  & \qquad  i \neq \ell  \\
   \log \left( \frac{Q_\ell \left(1-  \frac{\widehat{P}_{\ell}}{1-\alpha} \right)}{ \left(1 - Q_\ell \right)  \frac{\widehat{P}_{\ell}}{1-\alpha}  }\right)  &\qquad  i = \ell \ 
\end{cases} 
\end{align}
and 
\begin{align*}
\nu^* = \log \left( \frac{1-Q_\ell}{1-  \frac{\widehat{P}_{\ell}}{1-\alpha} }   \right) - 1
\end{align*}
where the Lagrangian \cite{boyd2009convex} is given as
\begin{align*}
L(P,\lambda,\nu) = \sum_i P_i \log \frac{P_i}{Q_i} + \sum_i \lambda_i \left(P_i - \frac{\widehat{P}_i}{1-\alpha}\right) + \nu \left(\sum_i P_i -1\right).
\end{align*}
These primal and dual optimal points are derived using methods similar to \cite{boyd2009convex} (p. 228, 248); in what follows, we simply verify the KKT conditions which suffice to complete the proof.   First, we confirm that the solution is a stationary point:
\begin{eqnarray*}
\left.
\frac{\partial L(P,\lambda,\nu)}{ \partial P_i} \right \vert_{P^*, \lambda^*, \nu^*} = \left. \log \frac{P_i}{Q_i} + 1 + \lambda_i + \nu \right \vert_{P^*, \lambda^*, \nu^*} = 0
\end{eqnarray*}
which holds for all $i$.   The complementary slackness condition is readily verified: 
\begin{align*}
\lambda_i^* \left(P_i^* - \frac{\widehat{P}_i}{1-\alpha} \right)  = 0, \quad i =1,\dots, n.
\end{align*}
It remains to show conditions under which the solution is primal and dual feasible.  First, $\lambda_\ell^* \geq 0$ provided 
\begin{align*}
\frac{Q_\ell \left(1-  \frac{\widehat{P}_{\ell}}{1-\alpha} \right)}{ \left(1 - Q_\ell \right)  \frac{\widehat{P}_{\ell}}{1-\alpha}  } \geq 1.
\end{align*}
After arranging terms, the above holds when $\alpha \leq 1 - \frac{\widehat{P}_\ell}{Q_\ell} =  \kappa(\widehat{P}||Q)$.
The primal equality constraint, $\sum_i P_i^* =1$, is readily verified.  Lastly, we check the primal inequality constraints. $P_\ell^*$ is trivially feasible.  For $i \neq \ell$, we require
\begin{align*}
P^*_i =  \frac{Q_i \left(1-  \frac{\widehat{P}_{\ell}}{1-\alpha} \right)}{ 1 - Q_\ell  }   \leq \frac{\widehat{P}_i}{1-\alpha}  
\end{align*}
which holds when 
\begin{align*}
\alpha \geq 1 - \widehat{P}_\ell -  \frac{\widehat{P}_i}{Q_i} \left(1 - Q_\ell  \right). 
\end{align*}
Since $\frac{\widehat{P}_i}{Q_i} \geq \frac{\widehat{P}_k}{Q_k} $ for all $i \neq \ell$, the solution is feasible if 
\begin{align*}
\alpha \geq 1 - \widehat{P}_\ell - \frac{\widehat{P}_k}{Q_k} \left(1 - Q_\ell   \right) .
\end{align*}
We conclude that the KKT conditions are satisfied for the range $\alpha$ specified in the statement of the theorem.  Since the objective is strictly convex the solution is unique, which completes the proof.  

\end{proof}

{\bf{Step 2:}}  We show that as $p$ approaches infinity, $\alpha$ approaches the separation distance.  More specifically, we have the following theorem. 
\begin{thm} \label{thm:largep}
\begin{align} \nonumber
\alpha_\mathrm{L} \leq  \kappa(\widehat{P}||Q^{0}) \qquad  \mbox{and} \qquad \lim_{p \rightarrow \infty} \alpha_\mathrm{L} = \kappa(\widehat{P}||Q^{0}) 
\end{align}
\end{thm}

We begin the proof by examining the behavior of $D_\alpha^*$ and $\alpha_{\mathrm{L}}$.  Note that $D_\alpha^*$ (the minimizer of (\ref{eqn:contop})) is monotone non-increasing in $\alpha$, as increasing $\alpha$ relaxes the constraints.  For $\alpha = \kappa(\widehat{P}||Q^{0})$, $D_\alpha^* = 0$ as the constraints allow $P_i=Q_i$ for all $i$ (as KL divergence is minimized if and only if  $P_i=Q_i$ for all $i$).   Define 
\begin{align}
 \gamma_\mathrm{L}(\alpha, p)  = \frac{1}{p(1-\alpha)} \log \left( \frac{1}{\epsilon} \right) + \frac{2n}{p(1-\alpha)} \log \left(p(1-\alpha)+1\right)    \hspace{-.3cm}   \nonumber
\end{align}
for $\alpha \in \lbrack 0,1 \rbrack$, $p>0$. We can write (\ref{eqn:alp_ub}) as
\begin{align}
\alpha_{\mathrm{L}} =  \max \left\{\alpha: D^*_\alpha \geq  \gamma_\mathrm{L}(\alpha, p)   \right\} \nonumber 
\end{align}
For fixed $p$, $ \gamma_\mathrm{L}(\alpha, p)$ is strictly increasing in $\alpha$ for $\alpha \in \lbrack 0,1 \rbrack$.    This (and since $D_\alpha^*$ is monotone non-decreasing in $\alpha$) implies existence and uniqueness of $\alpha_\mathrm{L}$ for fixed $p$.  Next, for fixed $\alpha$, $ \gamma_\mathrm{L}(\alpha, p)$ is strictly decreasing in $p$.  Since  $D_\alpha^*$ is not a function of $p$, we conclude that $\alpha_{\mathrm{L}}$ is non-decreasing in $p$.

Lastly, to prove the limit statement, we require $D_\alpha^*$ be left continuous at  $\alpha=\kappa(\widehat{P}||Q^{0})$; 
 for any $ \epsilon>0$, there exists some $\delta>0$ such that $D^*_{\kappa(\widehat{P}||Q^{0})-\delta} < \epsilon$.  
This follows as the objective is continuous in the optimization variables, and constraints are continuous in $\alpha$; an arbitrarily small increase in the objective can be realized by sufficiently reducing $\alpha$.

{\bf{Step 3:}} Bound  $\kappa(\widehat{P}||Q) - {\alpha_\mathrm{L}}$  using the closed form solution.

The value of KL divergence at $P^*$ from  (\ref{eqn:soon}) is 
\begin{eqnarray} \nonumber
 D_\alpha^* &=& D(P^*||Q) = \sum_{i\neq \ell} \frac{Q_i\left(1-  \frac{\widehat{P}_{\ell}}{1-\alpha} \right)}{ 1 - Q_\ell  }  \log \left( \frac{1-\frac{\widehat{P}_{\ell}}{1-\alpha}}{ 1 - Q_\ell  }   \right) + \frac{\widehat{P}_{\ell}}{1-\alpha} \log \frac{\frac{\widehat{P}_{\ell}}{1-\alpha} }{Q_\ell}  \\ \nonumber
&=&   \left(1- \frac{\widehat{P}_{\ell}}{1-\alpha} \right)  \log \left( \frac{1 - \frac{\widehat{P}_{\ell}}{1-\alpha} }{ 1 - Q_\ell  }   \right) + \frac{\widehat{P}_{\ell}}{1-\alpha} \log \frac{\frac{\widehat{P}_{\ell}}{1-\alpha}}{Q_\ell} \\ \nonumber
&\leq&   \frac{\left(Q_\ell- \frac{\widehat{P}_{\ell}}{1-\alpha} \right)^2}{Q_{\ell}(1-Q_\ell) }  \\
& = &  \frac{Q_\ell \left( \kappa(\widehat{P}||Q) - \alpha \right)^2}{(1-Q_\ell)(1-\alpha)^2 }  \label{eqn:ineq44}
\end{eqnarray}
where the inequality follows since $\log(x) \leq x - 1$.  We are ready to bound the difference between $\alpha_\mathrm{L}$ and the separation distance.  Recall the definition of  $\alpha_\mathrm{L}$;  $\alpha_\mathrm{L}$ must satisfy
\begin{eqnarray*}
 D_{\alpha_\mathrm{L}}^*\leq \frac{1}{p(1-{\alpha_\mathrm{L}})} \log \left( \frac{1}{\epsilon} \right)  + \frac{n}{p(1-{\alpha_\mathrm{L}})} \log \left(p(1-\alpha_\mathrm{L})+1\right)  
\end{eqnarray*}
 and by (\ref{eqn:ineq44})
 \begin{eqnarray*}
  \frac{Q_{\ell} \left( \kappa(\widehat{P}||Q) - {\alpha_\mathrm{L}}  \right)^2  }{(1- {\alpha_\mathrm{L}})(1-Q_\ell) } \leq \frac{1}{p} \log \left( \frac{1}{\epsilon} \right)  + \frac{n}{p} \log \left(p(1-{\alpha_\mathrm{L}})+1\right)  
\end{eqnarray*}
which implies the result
\begin{eqnarray*}
  \kappa(\widehat{P}||Q) - {\alpha_\mathrm{L}} \leq \sqrt{ \frac{1}{p} \log \left( \frac{1}{\epsilon} \right)  + \frac{n}{p} \log \left(p+1\right)   }  =
  O \left( \sqrt{\frac{\log p }{p} } \right).
\end{eqnarray*}

\end{document}